\documentclass[letterpaper, 10 pt, conference]{ieeeconf}  % Comment this line out if you need a4paper

\IEEEoverridecommandlockouts                              % This command is only needed if 
\title{\LARGE \bf
Signal Temporal Logic Task Decomposition via Convex Optimization*
}

\author{Maria Charitidou and Dimos V. Dimarogonas$^{1}$% <-this % stops a space
\thanks{*This work was partially supported by the Wallenberg AI, Autonomous Systems and Software Program (WASP) funded by the Knut and Alice Wallenberg Foundation, the ERC CoG LEAFHOUND and the Swedish Research Council.}% <-this % stops a space
\thanks{$^{1}$Both authors are with the Division of Decision and Control Systems, Royal Institute of Technology, 100 44 Stockholm, Sweden
        {\tt\small mariacha@kth.se} (M. Charitidou), {\tt\small dimos@kth.se} (D. V. Dimarogonas)}%
}
\usepackage{color}
\usepackage{comment}
\usepackage[backend=biber,style=ieee,sorting=none]{biblatex}
\usepackage{amssymb}
\usepackage{amsmath}

\usepackage{amsthm}
\usepackage{bm}
\allowdisplaybreaks
\theoremstyle{plain}
\newtheorem{theorem}{Theorem}
\newtheorem{remark}{Remark}
\newtheorem{assumption}{Assumption}
\newtheorem{proposition}{Proposition}
\newtheorem{problem}{Problem}

\newcommand{\xbf}{\mathbf{x}}
\newcommand{\ubf}{\mathbf{u}}
\newcommand{\zbf}{\mathbf{z}}
\newcommand{\nrob}{R}
\newcommand{\xset}{\mathbb{X}}
\newcommand{\uset}{\mathbb{U}}
\newcommand{\zset}{\mathcal{Z}}
\newcommand{\alw}{\mathcal{G}}
\newcommand{\ev}{\mathcal{F}}
\newcommand{\until}{\mathcal{U}}

\usepackage{graphicx} 
\usepackage{subcaption} % only subcaption with subfigures works great!!!

\begin{document}

\maketitle
\thispagestyle{empty}
\pagestyle{empty}

%%%%%%%%%%%%%%%%%%%%%%%%%%%%%%%%%%%%%%%%%%%%%%%%%%%%%%%%%%%%%%%%%%%%%%%%%%%%%%%%
\begin{abstract}
In this paper we focus on the problem of decomposing a global Signal Temporal Logic formula (STL) assigned to a multi-agent system to local STL tasks when the team of agents is a-priori decomposed to disjoint sub-teams. The predicate functions associated to the local tasks are parameterized as hypercubes depending on the states of the agents in a given sub-team. The parameters of the functions are, then, found as part of the solution of a convex program that aims implicitly at maximizing the volume of the zero level-set of the corresponding predicate function. Two alternative definitions of the local STL tasks are proposed and the satisfaction of the global STL formula is proven when the conjunction of the local STL tasks is satisfied. 
\end{abstract}

\section{Introduction}
Over the last decades, multi-agent systems have been considered in a variety of applications such as connectivity and formation control \cite{magnus}
or coverage \cite{coverage}. The complexity of  these applications has motivated the need of an expressive language, capable of describing complex task specifications for planning and control synthesis. 

Recently, extensive interest has been shown in planning under high-level task specifications expressed by Linear Temporal Logic (LTL) \cite{tl1,tl2}. In these methods the temporal formula, the environment and the agent dynamics are abstracted into  finite-transition systems. Then, graph-based methods are employed to find a discrete path satisfying the LTL specifications which is finally followed using continuous control laws. An important limitation of the aforementioned methods is the increasing computational complexity as the number of  the agents in the team becomes larger. Towards minimizing the computational costs, large effort has been devoted to the decomposition of a global LTL formula into local LTL tasks whose satisfaction depends on subsets of agents. Existing methods, applied to heterogeneous agents \cite{dimos2,ltl_services}, most often employ exhausting automata-based  approaches \cite{dimos1,dimos2,ltl_services} or more recently, cross-entropy optimization methods limited, though, to homogeneous agents \cite{ltldec_entropy}. 

All methods presented so far consider the satisfaction of LTL tasks without explicit time constraints. On the other hand, Signal Temporal Logic (STL) \cite{stl} can express complex tasks under strict deadlines. An advantage of STL over LTL, is the robust semantics \cite{stlrob, fainekos} it offers that allow the evaluation of the satisfaction of the task over a continuous-time signal, rendering the abstractions of the agents' dynamics obsolete. 

Existing methods for planning under STL  specifications consider a global STL formula and find plans as solutions to computationally prohibitive MILPs  \cite{mpc_raman,mpc_sadra} or to scalable convex programs \cite{lars2,kunal}. Other approaches compose \cite{decentralized} or assume the existence \cite{lars_linear} of local STL tasks whose satisfaction involves only a small subset of agents. This facilitates the design of decentralized frameworks that are inherently more robust to agents' failures and often cheaper in terms of communication. Towards decentralized control under global task specifications,  a satisfiability modulo theories (SMT) approach has been proposed in \cite{stldec2} for tasks described in caSTL, in which both the global formula and the team of agents is decomposed. Here, the decomposition is based on a set of services required for each task and a set of utility functions specifying the capabilities of the agents. Nevertheless, the decomposition of a global STL formula in continuous space and time remains an open problem.

In this paper we propose a novel framework for the decomposition of a global STL formula imposed on a multi-agent system into a set of local tasks when the team of agents is a-priori divided into disjoint sub-teams. The goal of the decomposition is to make the satisfaction of every local task dependent only to a subset of agents that belong to the same sub-team. Initially, the predicate functions corresponding to STL formulas forming the local tasks are parameterized as functions of the infinity norm of the agents' states while their parameters are found as part of the solution to a convex program that aims at maximizing the volume of their zero level-set. Although the choice of the parametric family of the predicate functions is not restrictive, our current choice allows us to draw conclusions on the volume of a continuous state-space set by incorporating a finite, but possibly large, number of constraints in the convex program. The number of these constraints differs per global STL task but depends solely on the number of the agents' states involved in its satisfaction. Two definitions of the local tasks that differ on the definition of the STL tasks originating from eventually formulas are introduced. Finally, for both definitions the satisfaction of the global STL formula is proven when the conjunction of the local tasks is satisfied.

The remainder of the paper is as follows: Section II includes the preliminaries and problem formulation. Section III introduces the proposed method for STL decomposition. Simulations are shown in Section IV and conclusions are summarized in Section V.

\section{Preliminaries and Problem Formulation}
The set of real and non-negative real numbers are denoted by $\mathbb{R}$ and $\mathbb{R}_{\geq 0}$ respectively. True and false are denoted by $\top, \bot$ respectively. Scalars and vectors are denoted by non-bold and bold letters respectively. The infinity norm of a vector $\xbf \in \mathbb{R}^n$ is defined as $\Vert \xbf \Vert_{\infty}=\max_i\vert \xbf_i \vert$, where $\xbf=\begin{bmatrix} \xbf_1 & \ldots & \xbf_n \end{bmatrix}^T$.
Given a finite set $ \mathcal{V}$, $\prod_{k\in \mathcal{V}}\xset_k$ denotes the Cartesian product of the sets $\xset_k, k\in  \mathcal{V}$. Given a rectangular matrix $A\in M_{n\times m}(\mathbb{R})$ we define the set $A\xset$ as $A\xset=\{A\xbf: \xbf \in \xset\}$. A square matrix $P\in M_n(\{0,1\})$ is called a \textit{permutation matrix} \cite[Ch. 0.9.5]{horn} if exactly
one entry in each row and column is equal to 1 and all other entries are 0. Consider the vectors $\xbf \in \mathbb{R}^n, \mathbf{y}\in \mathbb{R}^m$ with $n\leq m$ satisfying $\xbf=B\mathbf{y}$. The matrix $B=[b_{ij}]$ is called a \textit{selection matrix} if it has the following properties: 1) $b_{ij}\in \{0,1\}$, 2) $\sum_{j=1}^m b_{ij}=1, \forall i=1,\ldots,n$ and 3) $\sum_{i=1}^n b_{ij}=1, \forall j=1,\ldots,m$.

\subsection{Signal Temporal Logic (STL)}
Signal Temporal Logic (STL) determines whether a predicate $\mu$ is true or false. The validity of each predicate $\mu$ is evaluated based on a continuously differentiable function $h:\mathbb{R}^n \rightarrow \mathbb{R}$ as follows:
\begin{equation*}
    \mu=\begin{cases} \top, &h(\xbf) \geq 0 \\ \bot, & h(\xbf)< 0  \end{cases}
\end{equation*}
for $\xbf \in \mathbb{R}^n$. The basic STL formulas are given by the grammar:
$$ \phi:= \top \; | \;  \psi \;| \;\neg \phi \; | \; \phi_1 \land  \phi_2 \; |\; \alw_{[a,b]} \phi \;| \; \ev_{[a,b]} \phi \;| \; \phi_1 \; \until_{[a,b]} \;  \phi_2  $$
where $\phi_1, \phi_2  $ are STL formulas and $\alw_{[a,b]},\; \ev_{[a,b]}, \; \until_{[a,b]}$ is the always, eventually and until operator defined over the interval  $[a,b]$ with $0 \leq a \leq b$. Let $ \xbf \models \phi$ denote the satisfaction of the formula $\phi$ by a signal $\xbf:\mathbb{R}_{\geq 0} \rightarrow \mathbb{R}^n$. The formula $\phi$ is satisfiable if $\exists \; \xbf:\mathbb{R}_{\geq 0} \rightarrow \mathbb{R}^n$ such that $\xbf \models \phi$. The STL semantics for a signal $\xbf:\mathbb{R}_{\geq 0} \rightarrow \mathbb{R}^n$ are recursively given and can be found, e.g., in \cite{lars2}. STL is equipped with robustness metrics determining how robustly an STL formula $\phi$ is satisfied at time $t$ by a signal $\xbf$. These semantics are defined as follows \cite{stlrob, fainekos}: $\rho^{\mu}(\xbf,t)=h(\xbf(t))$, $\rho^{\neg \phi}(\xbf,t)=-\rho^{\phi}(\xbf,t)$, $ \rho^{\phi_1 \wedge \phi_2}(\xbf,t)=\min(\rho^{\phi_1}(\xbf,t),\rho^{\phi_2}(\xbf,t))$, $\rho^{\phi_1 \; \until_{[a,b]} \;  \phi_2}(\xbf,t)=\max_{t_1 \in [t+a,t+b]} \min(\rho^{\phi_2}(\xbf,t_1), \min_{t_2\in [t,t_1]} \rho^{\phi_1}(\xbf,t_2)) $, $\rho^{\ev_{[a,b]} \phi}(\xbf,t)=\max_{t_1\in [t+a,t+b]} \rho^{\phi}(\xbf,t_1) $, $\rho^{\alw_{[a,b]} \phi}(\xbf,t)=\min_{t_1\in [t+a,t+b]} \rho^{\phi}(\xbf,t_1)$. Finally, it should be noted that $\xbf \models \phi$ if $\rho^{\phi}(\xbf,0)>0$.

\subsection{Problem Formulation}

In this work we consider the following STL  fragment:
\begin{subequations}
\begin{align}
    \psi &:= \;  \mu \;| \;\neg \mu  \label{eq:f1} \\
    \varphi &:= \alw_{[a,b]} \psi \;| \; \ev_{[a,b]} \psi  \label{eq:f2}\\
    \phi&:=\bigwedge_{i=1}^{p} \varphi_i \label{eq:f3}
\end{align}
\end{subequations}
where $0\leq a \leq b < \infty$ and $p\geq 1$.
\begin{remark}
The STL fragment defined by \eqref{eq:f1}-\eqref{eq:f3} is expressive enough to accommodate until STL formulas of the form $\varphi=\psi_1 \until_{[a,b]} \psi_2$ where $\psi_i, i=1,2$ are defined by \eqref{eq:f1}. By definition, for any $t^* \in [a,b]$ the until formula $\varphi=\psi_1 \until_{[a,b]} \psi_2$ can be written as $\varphi=\alw_{[a,t^*]} \psi_1 \wedge \ev_{[t^*,t^*]} \psi_2$. Hence, if for a given time instant $t^* \in [a,b]$ there exists a signal $\xbf:\mathbb{R}_{\geq 0} \rightarrow \mathbb{R}^n$ such that $\xbf \models \big(\alw_{[a,t^*]} \psi_1 \wedge \ev_{[t^*,t^*]} \psi_2\big)$ then $\xbf \models \varphi$.
\end{remark}

Consider a team of $\nrob$ agents with each agent identified by its index $k \in \mathcal{V}=\{1, \ldots, \nrob\}$. For every agent  $k$ let $ \mathbf{x}_k \in \xset_k$ denote its state vector, where $\xset_k\subseteq \mathbb{R}^{\bar{n}_k} $ is a known, bounded, convex set for every $k\in \mathcal{V}$. Let $n=\sum_{k\in \mathcal{V}} \bar{n}_k$ and $\xbf=\begin{bmatrix} \xbf_1^T & \ldots & \xbf_R^T \end{bmatrix}^T \in \xset$ where $\xset=\prod_{k\in \mathcal{V}} \xset_k$ is convex as the Cartesian product of convex sets. Assume that the agents are decomposed in $v$ smaller teams $\{\mathcal{V}_1, \ldots, \mathcal{V}_v\}$, $\mathcal{V}_l \subseteq \mathcal{V}, l=1,\ldots,v$ that are disjoint, i.e., for any $l_1, l_2\in \{1,\ldots, v\}$ with $l_1\neq l_2$ it holds that $\mathcal{V}_{l_1}\cap \mathcal{V}_{l_2}=\emptyset$ and satisfy $\bigcup_{l=1}^v \mathcal{V}_l=\mathcal{V}$.

Consider a global STL formula $\phi$ of the form \eqref{eq:f3} with $\mathcal{I}= \{1,\ldots, p\}$ and sub-formulas $\varphi_i, \; i\in \mathcal{I}$ satisfying \eqref{eq:f1}-\eqref{eq:f2}. Let $[a_i,b_i]$ be the interval of satisfaction associated with the temporal operator of $\varphi_i, i\in \mathcal{I}$ and define the sets of always and eventually formulas of $\phi$ as $\mathcal{I}_{\alw}=\big\{i\in \mathcal{I}: \varphi_i=\alw_{[a_i,b_i]} \psi_i\big\}$ and $\mathcal{I}_{\ev}=\big\{i\in \mathcal{I}: \varphi_i=\ev_{[a_i,b_i]} \psi_i\big\}$ respectively. Observe that by definition of the STL fragment in \eqref{eq:f1}-\eqref{eq:f3} it holds that $\mathcal{I}=\mathcal{I}_{\alw} \cup \mathcal{I}_{\ev}$. Assume without loss of generality that the satisfaction of each $\varphi_i, i\in \mathcal{I}$ depends on multiple agents of different teams $\mathcal{V}_l$ and let $V_i\subseteq \{1,\ldots,v\}, i\in \mathcal{I}$ denote the set of indices of the agents' groups that have at least one member contributing to the satisfaction of $\varphi_i$. Since $\phi$ is a global task its satisfaction requires agents to be fully aware of the actions of their peers. However, in real-time scenarios communication between all agents may often be hard to establish, especially when the working environment of the agents is large. Addressing this problem, in this paper we propose decomposing the initial task $\phi$ into local tasks the satisfaction of which depends only on the agents in the same team $\mathcal{V}_l$. This problem is formally introduced as:

\begin{problem}
Given a global STL formula $\phi$ defined by \eqref{eq:f3} and the disjoint sets of agents $\mathcal{V}_l, l=1,\ldots,v$ satisfying $\bigcup_{l=1}^v \mathcal{V}_l=\mathcal{V}$ find  STL formulas $\phi_1, \ldots, \phi_v$ such that: 1) each STL formula $\phi_l$ depends on the agents in $\mathcal{V}_l$ and 2) $\xbf \models \big(\phi_1\wedge \ldots\wedge \phi_v\big) \Rightarrow \xbf \models \phi$ if such  $\xbf:\mathbb{R}_{\geq 0}\rightarrow \xset$ exists.
\end{problem}

\section{Decomposition of STL Formulas}

In this Section we design a number of STL tasks the satisfaction of which depends on a known subset of agents. Consider the formula $\phi$ defined by \eqref{eq:f3}. Let the predicate function $h_i:\xset \rightarrow \mathbb{R}$ associated with the formula $\varphi_i,i\in \mathcal{I}$. Then, the zero level-set of $h_i(\xbf)$ is defined as follows:
\begin{equation}
    \mathcal{S}_i=\{\xbf \in \xset: h_i(\xbf)\geq 0\} \label{eq:levelset}
\end{equation}
Here, we assume that $h_i(\xbf),i\in \mathcal{I}$ is a function whose value may depend on the states of all agents in $\mathcal{V}$. As a result guaranteeing the satisfaction of $\varphi_i, i\in \mathcal{I}$ may require the knowledge of all agents' actions and thus global communication. In real-time scenarios communication among all agents can become costly or hard due to packet losses or communications delays. On the other hand, decentralized approaches allow agents to communicate with a subset of their peers and optimize their actions with respect to a limited number of agents thus improving the computational complexity of the problem. 

In the context of STL control synthesis a decentralized approach involves the requirement of assigning to agents tasks whose satisfaction depends only to a subset of agents with established communication links , i.e., to $\mathcal{V}_l, l=1,\ldots,v$ while guaranteeing the satisfaction of the global task $\phi$. To that end, in this paper we propose  a set of STL tasks $\phi_l=\bigwedge_{q_i=1}^{p_l} \bar{\varphi}_{q_i}^l,\; l=1,\ldots,v$ whose satisfaction depends on the corresponding set of agents $\mathcal{V}_l$. Here, $\bar{\varphi}_{q_i}^l$ denotes the $q_i^l$-th formula of $\phi_l$ that is considered to be the result of the decomposition of the sub-formula $\varphi_i$ of \eqref{eq:f3}. If it is clear from context, we may omit the subscript of the index $q_i\in \{1,\ldots,p_l\}$. 

Let $\zbf_l \in \zset_l \subset \mathbb{R}^{n_l}$ be the states of the agents in $\mathcal{V}_l$ where $n_l=\sum_{k\in \mathcal{V}_l} \bar{n}_k$ and  $\zset_l=\prod_{k\in \mathcal{V}_l} \xset_k$. The vector $\zbf_l, l=1,\ldots,v$ can be obtained from $\xbf$ using the following equation:
\begin{equation}
    \zbf_l=E_l \xbf \label{eq:z2x}
\end{equation}
where $E_l \in M_{n_l\times n}(\{0,1\})$ is a selection matrix. Additionally, the vector $\xbf$ can be written with respect to the vectors $\zbf_l, l=1,\ldots,v$ as:
\begin{equation}
    \xbf=A \zbf \label{eq:permutation}
\end{equation}
where $\zbf=\begin{bmatrix} \zbf_1^T & \ldots & \zbf_v^T \end{bmatrix}^T$ and $A\in M_n(\{0,1\})$ is an appropriately chosen permutation matrix. Let $[a_{q}^l,b_{q}^l]$ and $h_{q}^l:\zset_l\rightarrow \mathbb{R}, q=1,\ldots,p_l, \; l=1,\ldots v$ denote the interval of satisfaction and predicate function corresponding to $\bar{\varphi}_{q}^l$ respectively. Here,  for every $l=1,\ldots,v$ we assume that $h_{q_i}^l(\zbf_l)=h_{q_i}^l(\zbf_l;\bm{\theta}_i^l), q_i=1,\ldots,p_l$ belongs to a known family of functions and its value depends on a set of parameters $\bm{\theta}_i^l\in \Theta_i^l \subseteq \mathbb{R}^{m_i^l}$ to be tuned towards maximizing the volume of the zero level-set  of $h_{q_i}^l(\zbf_l)$ defined as:
\begin{equation}
    S_{q_i}^l=\big\{\zbf_l \in \zset_l: h_{q_i}^l(\zbf_l)\geq 0 \big\} \label{eq:set}
\end{equation}

Based on the above we propose the following method for designing $\phi_l,l=1,\ldots,v$:
\begin{theorem}
Consider the global STL formula $\phi$ defined by \eqref{eq:f1}-\eqref{eq:f3} and the predicate function $h_i(\xbf)$ associated to $\varphi_i, i\in \mathcal{I}$. Assume that $\mathcal{S}_i\neq \emptyset$, where $\mathcal{S}_i, i\in \mathcal{I}$ is defined in \eqref{eq:levelset}. For every $i\in \mathcal{I}$ derive the functions $h_{q_i}^l(\zbf_l)$ as solutions to the following optimization problem:
\begin{subequations}\label{eq:dec}
\begin{align}
    \max_{\bm{\theta}_i^l\in \Theta_i^l, l\in V_i} \sum_{l\in V_i}\textit{vol}(S_{q_i}^l) \tag{\ref{eq:dec}}
    \end{align}
subject to:
\begin{align}
    \zbf_l&\in S_{q_i}^l, \quad l\in V_i\\
    \xbf&\in \mathcal{S}_i \label{eq:basiceq}\\
    \zbf_l&=E_l \xbf, \quad l\in V_i
\end{align}
\end{subequations}
where $\textit{vol}(S_{q_i}^l)$ denotes the volume of the set $S_{q_i}^l$ defined in \eqref{eq:set}. For every $l=1,\ldots,v$ define the formulas $\bar{\varphi}_{q_i}^l$ as follows:
\begin{equation}
    \bar{\varphi}_{q_i}^l=\begin{cases} \ev_{[a_{q_i}^l,b_{q_i}^l]} \bar{\mu}_{q_i}^l, \quad i \in \mathcal{I}_{\ev}\\\alw_{[a_{q_i}^l,b_{q_i}^l]} \bar{\mu}_{q_i}^l, \quad i\in \mathcal{I}_{\alw}
    \end{cases} \label{eq:newformula}
\end{equation}
with 
\begin{subequations}
\begin{align}
    [a_{q_i}^l,b_{q_i}^l]&=\begin{cases}[t_i,t_i], \quad i \in \mathcal{I}_{\ev}\\ [a_i,b_i],\quad i\in \mathcal{I}_{\alw} \end{cases} \label{eq:interval}\\
    \bar{\mu}_{q_i}^l&=\begin{cases} \top, & h_{q_i}^l(\zbf_l)\geq 0 \\ \bot, & h_{q_i}^l(\zbf_l)< 0  \end{cases}  \label{eq:predicate}
\end{align}
\end{subequations}
where $\mathcal{I}=\mathcal{I}_{\alw} \cup \mathcal{I}_{\ev}$, $t_i\in [a_i,b_i]$ and $[a_i,b_i]$ is the interval of satisfaction associated with each $\varphi_i$ of the global formula $\phi$. Let $\phi_l=\bigwedge_{q_i=1}^{p_l} \bar{\varphi}_{q_i}^l$, $l=1,\ldots,v$. If there exists $\xbf:\mathbb{R}_{\geq 0} \rightarrow \xset$ such that $\rho^{\phi_1\wedge \ldots\wedge \phi_v}(\xbf,0)>0$, then $\rho^\phi(\xbf,0)>0$.
\end{theorem}

\begin{proof}
For every $i\in \mathcal{I}$, \eqref{eq:dec} aims at maximizing the volume of the $S_{q_i}^l,l\in V_i$ which  underapproximates the projection set of $\mathcal{S}_i$ onto $\zset_l$. Since $\mathcal{S}_i\neq \emptyset$ for every $i\in \mathcal{I}$,  \eqref{eq:dec} is always feasible.
By definition of the robust semantics and the definition of the min operator it holds that:
\begin{equation*}
 \rho^{\phi_1\wedge \ldots\wedge \phi_v}(\xbf,0)\leq   \rho^{\phi_l}(\xbf,0), \; l=1,\ldots,v
\end{equation*}
As a result if there exists $\xbf:\mathbb{R}_{\geq 0} \rightarrow \xset$ such that $\rho^{\phi_1\wedge \ldots\wedge \phi_v}(\xbf,0)>0$ then $\rho^{\phi_l}(\xbf,0)> 0$ for every $l=1,\ldots,v$. By design, the satisfaction of $\phi_l$ depends on a subset of agents, thus $\rho^{\phi_l}(\xbf,0)=\rho^{\phi_l}(\zbf_l,0)> 0$ where $\zbf_l$ satisfies \eqref{eq:z2x}. Then, by the definition of the robust semantics for every $l=1,\ldots,v$ and $ q_i=1,\ldots,p_l$ it holds that:
\begin{equation*}
   0< \rho^{\phi_l}(\zbf_l,0)=\min_{q_i=1,\ldots,p_l}   \rho^{\bar{\varphi}_{q_i}^l}(\zbf_l,0)\leq   \rho^{\bar{\varphi}_{q_i}^l}(\zbf_l,0)
\end{equation*}
If $i\in \mathcal{I}_{\alw}$, then $\varphi_{q_i}^l$ is an always formula. Hence due to \eqref{eq:newformula}, $ \rho^{\bar{\varphi}_{q_i}^l}(\zbf_l,0)>0$ implies $h_{q_i}^l(\zbf_l(t))> 0$ for every $t\in [a_i,b_i]$, $l\in V_i$. Since $h_{q_i}^l(\zbf_l), l\in V_i$ is a feasible solution of \eqref{eq:dec}, it holds that $h_i(\xbf(t))>0, \forall t\in [a_i,b_i]$ where $\xbf(t)=A\zbf(t)$ and $\zbf(t)=\begin{bmatrix}\zbf_1^T(t) &\ldots & \zbf_v^T(t) \end{bmatrix}^T$. Hence, $\rho^{\varphi_i}(\xbf,0)>0$. If $i\in \mathcal{I}_{\ev}$, then due to \eqref{eq:newformula} and \eqref{eq:interval}-\eqref{eq:predicate}, for every $l\in V_i$ it holds that: $ \rho^{\bar{\varphi}_{q_i}^l}(\zbf_l,0)=h_{q_i}^l(\zbf_l(t_i))>0$. Following a similar argument as before, we can conclude that $h_i(\xbf(t_i))>0$ where $\xbf(t_i)=A\zbf(t_i)$. This implies that $\max_{t\in [a_i,b_i]}h_i(\xbf(t))\geq h_i(\xbf(t_i))>0$ leading to $\rho^{\varphi_i}(\xbf,0)>0$. Then, the result follows by the fact that $\rho^{\phi}(\xbf,0)=\min\big( \min_{i\in \mathcal{I}_{\alw}} \rho^{\varphi_i}(\xbf,0), \min_{i\in \mathcal{I}_{\ev}} \rho^{\varphi_i}(\xbf,0)\big)$.
\end{proof}

In problem \eqref{eq:dec} the goal is to maximize the volume of set $S_{q_i}^l$ by exhaustively evaluating $h_{q_i}^l(\zbf_l)$ over the continuous set $\zset_l$, which is in practice intractable. Another limitation of the proposed problem is often the lack of a known formula for computing the volume of a set, unless $h_{q_i}^l(\zbf_l)$ belongs to a specific class of functions such as the class of ellipsoids.

Aiming at reducing the computational complexity of the STL decomposition problem described above, we propose a convex formulation for designing the predicate functions corresponding to \eqref{eq:newformula} for every $l=1,\ldots,v$. The computational benefits of the proposed approach are related to the number of points in $\zset_l, l\in V_i$ that are considered for evaluation of the satisfaction of \eqref{eq:basiceq}. More specifically, contrary to \eqref{eq:dec}, in this approach only a finite number of points is evaluated that depends on the number of states $\xbf_k$ of the agents $k\in \mathcal{V}_l, l\in V_i$ involved in the satisfaction of $h_i(\xbf)$. Let $d_i^l\geq 1$ be the number of states in $\zbf_l, l\in V_i$ contributing to $h_i(\xbf)$. Since the global formula is a-priori given, the elements of $\zbf_l , l\in V_i$ on which the predicate function $h_i(\xbf)$ depends are known. Hence, we may write $h_i(\xbf),i \in \mathcal{I}$ as:
\begin{subequations}
\begin{align}
 h_i(\xbf)&=h_i(\mathbf{y}_{\alpha(1)},\ldots,\mathbf{y}_{\alpha(\vert V_i \vert)})  \label{eq:dependency2}\\
\mathbf{y}_{\alpha(c)}&=B_i^{\alpha(c)} \zbf_{\alpha(c)}, \quad c=1,\ldots, \vert V_i \vert    \label{eq:dependency}
\end{align}
\end{subequations}
where $\alpha: \{1,\ldots,\vert V_i \vert\} \rightarrow V_i$ is an injective function defined as $\alpha(c)=l$ and where $B_i^{\alpha(c)}\in M_{d_i^l\times n_l}(\{0,1\})$ is an appropriate selection matrix and $\zbf_{\alpha(c)} \in \zset_{\alpha(c)}$.

Based on the above, we can consider a special class of concave functions of the following form:
\begin{equation}
    h_{q_i}^l(\zbf_l)=r_{q_i}^l-\Vert B_i^l(\zbf_l-\mathbf{c}_{q_i}^l) \Vert_{\infty}, \quad q_i=1,\ldots,p_l \label{eq:prinf}
\end{equation}
where $r_{q_i}^l\in \mathbb{R}_{\geq 0}$, $\mathbf{c}_{q_i}^l\in \zset_l$ and $B_i^l\in M_{d_i^l\times n_l}(\{0,1\})$ is the same selection matrix considered in \eqref{eq:dependency} with $\alpha(c)=l$.
Let $J_{q_i}^l\subseteq \{1,\ldots,n_l\}$ denote the set of indices of the columns of  $B_i^l$ with non-zero entries. Given the predicate functions defined by \eqref{eq:prinf}, it follows that:
\begin{equation}
   h_{q_i}^l(\zbf_l)\geq 0 \Leftrightarrow  \quad \zbf_l(\eta) \in [-r_{q_i}^l+\mathbf{c}_{q_i}^l(\eta),r_{q_i}^l+\mathbf{c}_{q_i}^l(\eta)]
\end{equation}
for every $\eta \in J_{q_i}^l$ where $\zbf_l(\eta),\mathbf{c}_{q_i}^l(\eta)$ denote the $\eta$-th element of the vectors $\zbf_l,\mathbf{c}_{q_i}^l$ respectively. For every $i\in\mathcal{I}$ and $l\in\{1,\ldots,v\}$ consider the following set of vectors:
\begin{equation}
\begin{split}
  \mathcal{P}_i^l=\big\{\bm{\xi}\in \zset_l: \bm{\xi}(\eta)&=-r_{q_i}^l+\mathbf{c}_{q_i}^l(\eta) \; \text{or} \\ \bm{\xi}(\eta)&=r_{q_i}^l+\mathbf{c}_{q_i}^l(\eta), \eta \in J_{q_i}^l \big\} \label{eq:vertexset}
  \end{split}
\end{equation}
where $\bm{\xi}(\eta)$ denotes the $\eta$-th element of $\bm{\xi}$. If $r_{q_i}^l\geq 0$, the set $B_i^l\mathcal{P}_i^l$ consists of the vertices of a hypercube in $\mathbb{R}^{d_i^l}$ of edge length $r_{q_i}^l$ and center $\mathbf{c}_{q_i}^l$.
Hence, its cardinality will be equal to $2^{d_i^l}$. To guarantee the convexity of the proposed problem we pose the following assumption:

\begin{assumption}
For every $i\in \mathcal{I}$ the predicate function $h_i(\xbf)$ is concave in $\xset$.
\end{assumption}

%Based on the above we may state the following:
\begin{theorem}
Consider the global STL formula $\phi$ defined by \eqref{eq:f1}-\eqref{eq:f3} and the predicate functions $h_i(\xbf), i\in \mathcal{I}$ associated to $\varphi_i$. Let Assumption 1 hold. For every $i\in \mathcal{I}$ assume that $\mathcal{S}_i\neq \emptyset$, where $\mathcal{S}_i$ is defined in \eqref{eq:levelset}. Consider the functions $ h_{q_i}^l(\zbf_l), \; q_i=1,\ldots,p_l, \; l=1,\ldots,v$ defined by \eqref{eq:prinf} where $\mathbf{c}_{q_i}^l,r_{q_i}^l$ are parameters found as the solution to the following optimization problem:
\begin{subequations}\label{eq:convex}
\begin{align}
    \max_{\mathbf{c}_{q_i}^l,r_{q_i}^l} \sum_{l\in V_i} r_{q_i}^l  \tag{\ref{eq:convex}}
    \end{align}
subject to:
\begin{align}
  h_i(\mathbf{y}_{\alpha(1)},\ldots,\mathbf{y}_{\alpha(\vert V_i \vert)}) &\geq 0 \\
 \mathbf{y}_{\alpha(c)}&\in B_i^{\alpha(c)}\mathcal{P}_i^{\alpha(c)}, \quad c=1,\ldots,\vert V_i \vert
\end{align}
\end{subequations}
where $\alpha: \{1,\ldots,\vert V_i \vert\} \rightarrow V_i$  and $B_i^{\alpha(c)}\in M_{d_i^l\times n_l}(\{0,1\})$ is the injective function and selection matrix respectively considered in \eqref{eq:dependency2}-\eqref{eq:dependency} and $\mathcal{P}_i^{\alpha(c)}$ is the set defined by \eqref{eq:vertexset} for every $\alpha(c)=l\in V_i$.
For every $l=1,\ldots,v$ define the formulas $\bar{\varphi}_{q_i}^l$ based on \eqref{eq:newformula} and \eqref{eq:interval}-\eqref{eq:predicate}
and consider the decomposed STL formulas $\phi_l=\bigwedge_{q_i=1}^{p_l} \bar{\varphi}_{q_i}^l$, $l=1,\ldots,v$. If there exists $\xbf:\mathbb{R}_{\geq 0} \rightarrow \xset$ such that $\rho^{\phi_1\wedge \ldots\wedge \phi_v}(\xbf,0)>0$, then $\rho^\phi(\xbf,0)>0$.
\end{theorem}
\begin{proof}
For every $i\in \mathcal{I}$, \eqref{eq:convex} finds the maximum volume sets $B_i^lS_{q_i}^l,l\in V_i$ which are underapproximations of the projection sets of $\mathcal{S}_i$ onto $B_i^l\zset_l$. Since $\mathcal{S}_i\neq \emptyset$ for every $i\in \mathcal{I}$,  \eqref{eq:convex} is always feasible. To simplify notation let the sets $W=\big\{\mathbf{y}: \mathbf{y}_{l}=B_i^{l}\zbf_{l}, \; \zbf_l \in S_{q_i}^l, \; l\in V_i\big\}$ and $W^\prime=\big\{\mathbf{y}:      \mathbf{y}_{l}=B_i^{l}\bm{\xi}_l, \; \bm{\xi}_l\in \mathcal{P}_i^l, \; l\in V_i\big\}$ where $\mathbf{y}=\begin{bmatrix} \mathbf{y}_{\alpha(1)}^T & \ldots & \mathbf{y}_{\alpha(\vert V_i \vert)}^T \end{bmatrix}^T$. The sets $B_i^lS_{q_i}^l$ are convex since they are projection sets of the zero-level sets of the concave function $h_{q_i}^l(\zbf_l)$ defined by \eqref{eq:prinf}. Hence, $W$ is convex as the Castesian product of convex sets. By Caratheodory's theorem \cite[Th. 17.1]{rockafellar} every point $\mathbf{y}\in W$ can be written as a convex combination of $d_i+1$ points where $d_i=dim(W)=\sum_{l\in V_i} d_i^l$. Observe that $W^\prime \subset W$ with $\vert W^\prime \vert= 2^{d_i}> d_i$. Applying Caratheodory's theorem, we write any point $\mathbf{y}\in W$ as a convex combination of the form: $\mathbf{y}=\sum_{j=1}^{d_i+1} \lambda_j \mathbf{y}_j^\prime$ where $\mathbf{y}_j^\prime \in W^\prime, \; \lambda_j\geq 0 $ and  $ \sum_{j=1}^{d_i+1}\lambda_j=1$. By feasibility of \eqref{eq:convex} and due to Assumption 1 we can conclude that $h_i(\mathbf{y}_{\alpha(1)},\ldots,\mathbf{y}_{\alpha(\vert V_i \vert)})\geq 0$ for any $\mathbf{y}\in W$ with $\mathbf{y}=\begin{bmatrix} \mathbf{y}_{\alpha(1)}^T & \ldots & \mathbf{y}_{\alpha(\vert V_i \vert)}^T \end{bmatrix}^T$. The rest of the proof is similar to that of Theorem 1.
\end{proof}

For $i\in \mathcal{I}_{\ev}$ the new STL tasks, defined by \eqref{eq:newformula}, are expected to be satisfied at a specific time instant $t_i\in [a_i,b_i]$ which is considered a designer's choice. However, in many cases pre-determining the time instant of satisfaction of a formula may lead to conservatism and reduced performance. An alternative would be to allow satisfaction of the local formulas over time intervals $[a_q^l,b_q^l]\subseteq [a_i,b_i]$. Then, in order to guarantee the satisfaction of the global formula we can define the local tasks corresponding to  $\varphi_i, i\in \mathcal{I}_{\ev}$ as STL tasks of the form $\alw_{[a_{q_i}^l,b_{q_i}^l]} \mu$.
This is depicted in the following Proposition:
\begin{figure*}[!t]
     \centering
     \begin{subfigure}[b]{0.32\textwidth}
         \centering
       \includegraphics[width=\textwidth]{ 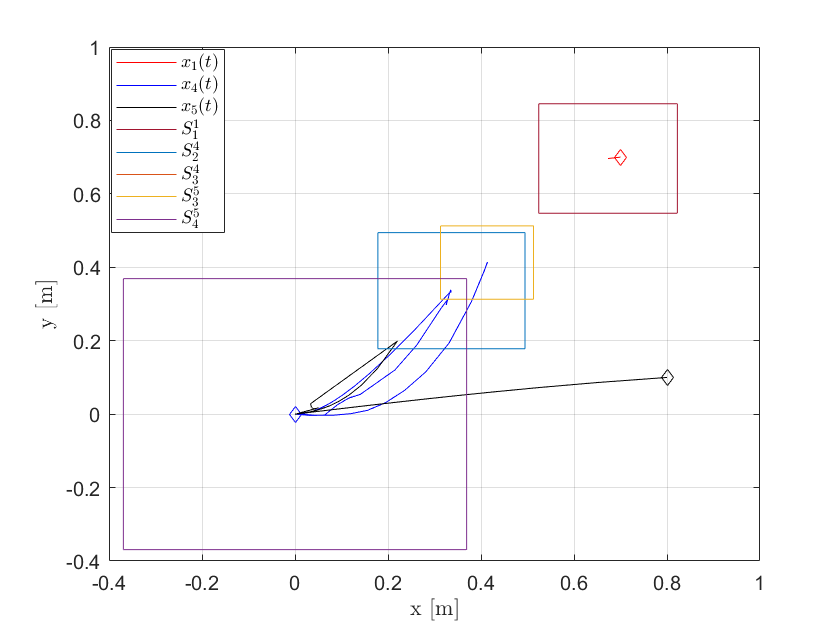}
         \caption{Trajectories of agents 1,4,5}
         \label{fig:ev1}
     \end{subfigure}
     \hfill
     \begin{subfigure}[b]{0.32\textwidth}
         \centering
         \includegraphics[width=\textwidth]{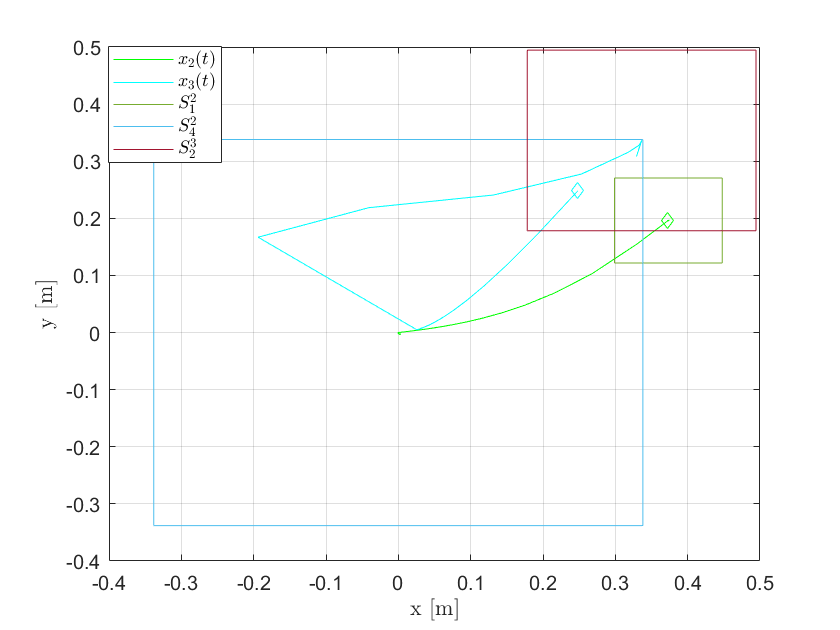}
         \caption{Trajectories of agents 2,3}
         \label{fig:ev2}
     \end{subfigure}
     \hfill
     \begin{subfigure}[b]{0.32\textwidth}
         \centering
         \includegraphics[width=\textwidth]{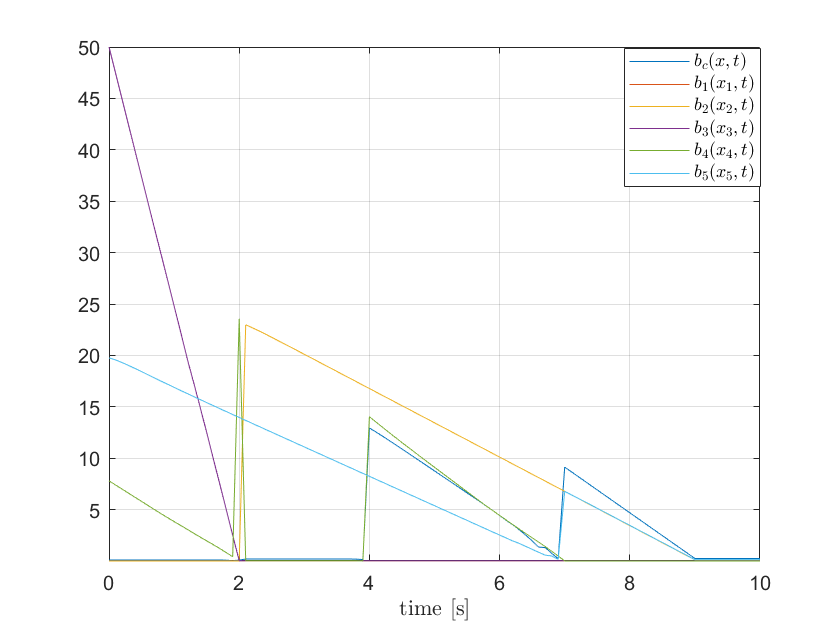}
         \caption{Barrier Function Evolution}
         \label{fig:ev3}
     \end{subfigure}
        \caption{Agents' Trajectories under the local STL tasks defined based on \eqref{eq:newformula}, \eqref{eq:interval}-\eqref{eq:predicate} and Barrier Function Evolution}
        \label{fig:evall}
\end{figure*}
\begin{proposition}
Consider the global STL formula $\phi$ defined by \eqref{eq:f1}-\eqref{eq:f3}. Let Assumption 1 hold. For every $i\in \mathcal{I}$ assume that $\mathcal{S}_i\neq \emptyset$, where $\mathcal{S}_i$ is defined by \eqref{eq:levelset}. For every $i\in \mathcal{I}$ consider the functions $h_{q_i}^l(\zbf_l), l\in V_i$ defined by \eqref{eq:prinf} with their parameters found as solutions  to \eqref{eq:convex}. Let the STL formula $\bar{\varphi}_{q_i}^l$ be defined as: 
\begin{equation}
    \bar{\varphi}_{q_i}^l= \alw_{[a_{q_i}^l,b_{q_i}^l]} \bar{\mu}_{q_i}^l \label{eq:always}
\end{equation}
where
\begin{equation}
    [a_{q_i}^l,b_{q_i}^l]\begin{cases}\subseteq[a_i,b_i], \quad i \in \mathcal{I}_{\ev}\\= [a_i,b_i],\quad i\in \mathcal{I}_{\alw} \end{cases} \label{eq:alwaysint}
\end{equation}
and $\bar{\mu}_{q_i}^l$ is a predicate defined by \eqref{eq:predicate}. Let $\phi_l=\bigwedge_{q_i=1}^{p_l} \bar{\varphi}_{q_i}^l$, $l=1,\ldots,v$. If there exists $\xbf:\mathbb{R}_{\geq 0} \rightarrow \xset$ such that $\rho^{\phi_1\wedge \ldots\wedge \phi_v}(\xbf,0)>0$, then $\rho^\phi(\xbf,0)>0$.
\end{proposition}

\begin{proof}
For $i\in \mathcal{I}_{\alw}$ the proof follows similar arguments to Theorem 1. For $i\in \mathcal{I}_{\ev}$, if  $\rho^{\phi_1\wedge \ldots\wedge \phi_v}(\xbf,0)>0$ then, by \eqref{eq:always}-\eqref{eq:alwaysint} and the definition of the robust semantics,  $\rho^{\bar{\varphi}_{q_i}^l}(\zbf_l,0)>0$ implies $ h_{q_i}^l(\zbf_l(t))>0$ for every $ t\in [a_{q_i}^l,b_{q_i}^l]$ and $l\in V_i$. Since $ h_{q_i}^l(\zbf_l), l \in V_i$ are feasible solutions of \eqref{eq:convex} we may conclude that $h_i(\xbf(t))>0$, where $\xbf(t)=A\zbf(t)$ and $\zbf(t)=\begin{bmatrix}\zbf_1^T(t) &\ldots & \zbf_v^T(t) \end{bmatrix}^T$, for every $t\in [a_{q_i}^l,b_{q_i}^l]\subseteq [a_i,b_i]$. Hence, $\rho^{\varphi_i}(\xbf,0)=\max_{t\in [a_i,b_i]} h_i(\xbf(t))\geq \max_{t\in [a_{q_i}^l,b_{q_i}^l]} h_i(\xbf(t))>0$. The rest of the proof is similar to that of Theorem 1.
\end{proof}

\begin{figure*}[!t]
     \centering
     \begin{subfigure}[b]{0.32\textwidth}
         \centering
       \includegraphics[width=\textwidth]{ 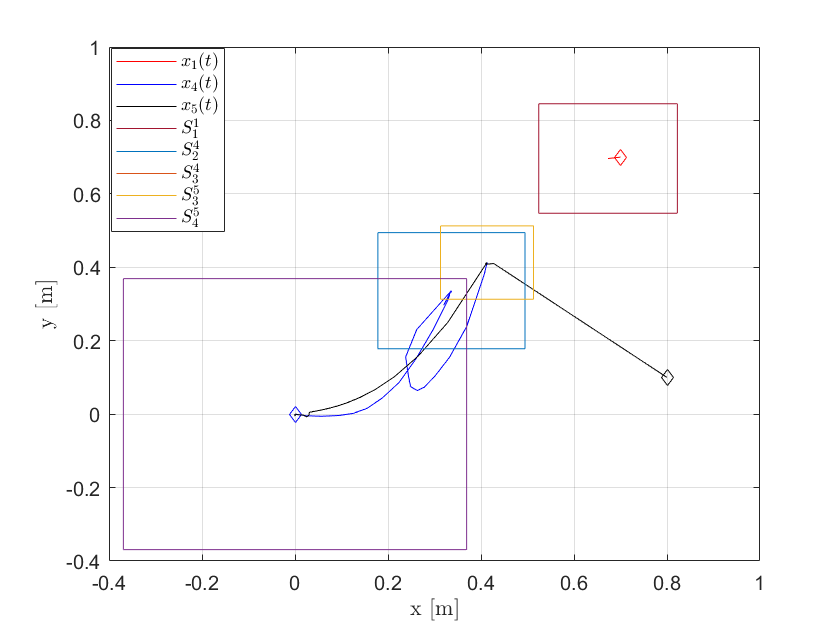}
         \caption{Trajectories of agents 1,4,5}
         \label{fig:alw1}
     \end{subfigure}
     \hfill
     \begin{subfigure}[b]{0.32\textwidth}
         \centering
         \includegraphics[width=\textwidth]{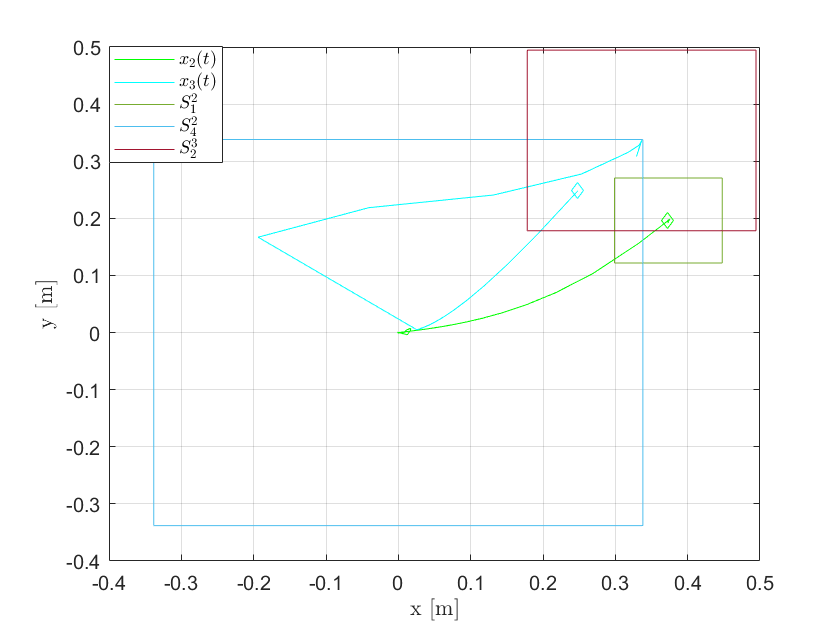}
         \caption{Trajectories of agents 2,3}
         \label{fig:alw2}
     \end{subfigure}
     \hfill
     \begin{subfigure}[b]{0.32\textwidth}
         \centering
         \includegraphics[width=\textwidth]{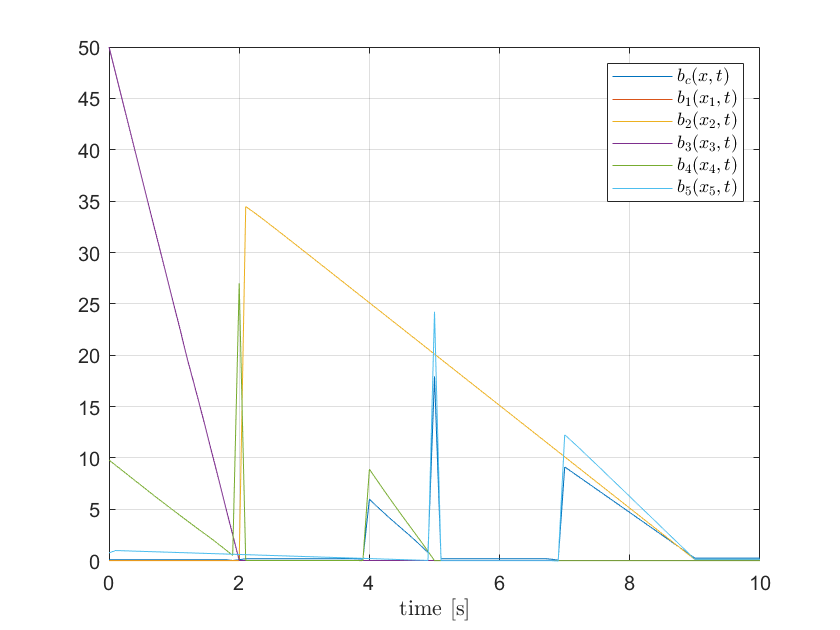}
         \caption{Barrier Function Evolution}
         \label{fig:alw3}
     \end{subfigure}
        \caption{Agents' Trajectories under the local STL tasks defined based on \eqref{eq:predicate}, \eqref{eq:always} and \eqref{eq:alwaysint} and Barrier Function Evolution}
        \label{fig:alwall}
\end{figure*}
\section{Simulations}
Consider a team of $R=5$ agents. Without loss of generality the team is decomposed in 5 sub-teams: $\mathcal{V}_k=\{k\}, k\in \mathcal{V}$. The agents' states $\xbf_k, k\in \mathcal{V}$ evolve over time based on the following equation:
\begin{equation*}
  \dot{\xbf}_k=A_k\xbf_k+\ubf_k , \; k=1,\ldots,5
\end{equation*}
where $A_k=\begin{bmatrix}-0.5 &0\\1 & -1 \end{bmatrix}$ for every $ k\in \{1,2,5\}$ and $A_k=\begin{bmatrix}-1 &-1\\0 & -3 \end{bmatrix}$ for $k\in \{3,4\}$. The states and inputs of the agents are subject to constraints, i.e., $\xbf_k \in \xset$, $\ubf_k \in \uset$ where $\xset=\{\xbf\in \mathbb{R}^2: \Vert \xbf \Vert_2 \leq d_x\}$, $\uset=\{\ubf\in \mathbb{R}^2: \Vert \ubf \Vert_2 \leq d_u\}$, $d_x=1$ and $d_u=5$. Consider the global STL formula $\phi=\bigwedge_{i=1}^4 \varphi_i$ where $\varphi_i,i\in \mathcal{I}$ are defined as: $ \varphi_1=\alw_{[0,2.1]}( \Vert \xbf_1-\xbf_2-p_x \Vert_2^2 \leq 0.1)$, $\varphi_2= \alw_{[2,4]}( \Vert \xbf_3-\xbf_4 \Vert_{2}^2 \leq 0.2)$, $  \varphi_3= \ev_{[3,7]}(\Vert \xbf_5-\xbf_4\Vert_{P_1}^2 \leq 0.2)$ and $\varphi_4= \ev_{[8,10]}(\Vert \xbf_5-\xbf_2\Vert_{P_2}^2 \leq 0.25)$, where $p_x=\begin{bmatrix} 0.3 & 0.5\end{bmatrix}^T$ and $P_1=\text{diag}(4,1), \; P_2=\text{diag}(0.1,0.4)$ are positive definite weight matrices. Since the predicate functions corresponding to $\varphi_i, i\in \mathcal{I}$ are quadratic, the proposed problem \eqref{eq:convex} becomes a Quadratically Constrained Quadratic Program (QCQP) and is efficiently solved using 
%the SCIP solver \cite{scip} available via
the \textit{Opti Toolbox} \cite{opti}. The average computational time of the QCQPs is 0.052sec on an Intel Core i7-8665U with 16GB RAM using MATLAB. 

To verify the validity of Theorem 2 and Proposition 1 we design agents' trajectories using the MPC scheme  proposed in \cite{ecc} with a sampling frequency of 10 Hz and optimization horizon length $N=1$. Each agent $k$ solves a local MPC problem without communicating with its peers since the satisfaction of the assigned tasks depends only on its own behavior. Here, a single, time-varying barrier $b_k(\xbf_k,t), k\in \mathcal{V}$ is considered and designed offline encoding the local STL task  specifications $\phi_k$ corresponding to $\mathcal{V}_k$. For every subtask of $\phi_k$ a temporal behavior is designed for agent $k$ such that the satisfaction of $\phi_k$ with a robustness value $r=0.005$ is guaranteed when $b_k(\xbf_k,t)\geq 0$ is true for every $t\in [0,10]$. For details on the design of the barrier function $b_k(\xbf_k,t)$ see \cite{lars_linear,ecc}. The local STL task $\phi_k$ assigned to each agent $k$ is defined by \eqref{eq:newformula}, \eqref{eq:interval}-\eqref{eq:predicate} as follows:
\begin{align*}
    \phi_1&=\alw_{[0,2.1]}\;\bar{\mu}_1^1=\bar{\varphi}_1^1\\
    \phi_2&= (\alw_{[0,2.1]}\;\bar{\mu}_1^2)\wedge(\ev_{[9,9]}\; \bar{\mu}_4^2)=\bar{\varphi}_1^2 \wedge \bar{\varphi}_4^2\\
    \phi_3&= \alw_{[2,4]}\;\bar{\mu}_2^3=\bar{\varphi}_2^3 \\
    \phi_4&= (\alw_{[2,4]}\;\bar{\mu}_2^4)\wedge(\ev_{[7,7]}\; \bar{\mu}_3^4)=\bar{\varphi}_2^4 \wedge \bar{\varphi}_3^4\\
    \phi_5&= (\ev_{[7,7]}\;\bar{\mu}_3^5)\wedge(\ev_{[9,9]}\; \bar{\mu}_4^5)=\bar{\varphi}_3^5 \wedge \bar{\varphi}_4^5
\end{align*}
In Figure \eqref{fig:ev1}, \eqref{fig:ev2} the agents' trajectories and the zero level sets $S_q^k$ of the predicate functions $h_q^k(\xbf_k)$ are shown when the parameters $c_q^k,r_q^k$ are found as solutions to \eqref{eq:convex}. Since the agents move on $\mathbb{R}^2$ and $r_q^k\neq 0$ for every $q$ and $k$ the zero level sets define square areas with edge length $r_q^k$. In Figure \eqref{fig:ev3} the evolution of the local barrier functions $b_k(\xbf_k,t)$ is shown. Since $\min_k \inf_{t\in [0,10]} b_k(\xbf_k,t)\geq 5.38\cdot 10^{-4}$ is true, we can conclude that $\rho^{\phi_k}(\xbf_k,0)\geq  0.005$ for every $k\in \mathcal{V}$. To validate Theorem 1 and given the trajectories of the agents found by the local MPC controllers we aim at designing a barrier function $b_c(\xbf,t)$ encoding the global specifications described by $\phi$ and evaluating its value over the interval $[0,10]$. If $b_c(\xbf,t)\geq 0$ is true for every $t\in [0,10]$, then the global formula $\phi$ is satisfied. From Figure \eqref{fig:ev3} we have that $\inf_{t\in [0,10]} b_c(\xbf,t)\geq 0.0234$. Hence, $\xbf \models \phi$.

Next, we consider the alternative definition of the local tasks as described in Proposition 1. Observe that the local tasks $\phi_1, \phi_3$ remain the same. The new local tasks $\phi_2, \phi_4, \phi_5$ are defined as: $\phi_2=\bar{\varphi}_1^2 \wedge\bar{\varphi}_4^2$, $\phi_4=\bar{\varphi}_2^4 \wedge \bar{\varphi}_3^4$ and $\phi_5=\bar{\varphi}_3^5 \wedge \bar{\varphi}_4^5 $, where $\bar{\varphi}_1^2=\alw_{[0,2.1]}\;\bar{\mu}_1^2$, $\bar{\varphi}_4^2=\alw_{[9,10]}\; \bar{\mu}_4^2$, $\bar{\varphi}_2^4=\alw_{[2,4]}\;\bar{\mu}_2^4$, $\bar{\varphi}_3^4=\alw_{[5,7]}\; \bar{\mu}_3^4$, $\bar{\varphi}_3^5=\alw_{[5,7]}\;\bar{\mu}_3^5$ and $\bar{\varphi}_4^5=\alw_{[9,10]}\; \bar{\mu}_4^5$. In Figure \eqref{fig:alw1} and \eqref{fig:alw2} the agents' trajectories are shown. Following a similar procedure as before, we design a set of local barrier functions $b_k(\xbf_k,t)$ and a function $b_c(\xbf,t)$ with robustness $r=0.005$. Based on Figure \eqref{fig:alw3}, $\min_k \inf_{t\in [0,10]} b_k(\xbf_k,t)\geq 5.17\cdot 10^{-4}$ implying $\xbf_k \models \phi_k, k \in \mathcal{V}$. Additionally, it holds that $\inf_{t\in [0,10]} b_c(\xbf,t)\geq 0.0234$. Hence, $\rho^{\phi}(\xbf,0)\geq 0.005$.

\section{Conclusions}
In this work a global STL formula is decomposed to a set of local STL tasks whose satisfaction depends on an a-priori chosen subset of agents. The predicate functions of the new formulas are chosen as functions of the infinity norm of the agents' states. A convex optimization problem is, then, designed for  optimizing  their parameters towards increasing the volume of their zero level-sets. Two alternatives are proposed for defining the local STL tasks in both of which the interval of satisfaction corresponding to the eventually formulas is considered a designer's choice. Future work will consider a more sophisticated framework for choosing the interval of satisfaction of the formulas aiming at increasing the total robustness of the task.

%\addtolength{\textheight}{-12cm}   % This command serves to balance the column lengths
                                  % on the last page of the document manually. It shortens
                                  % the textheight of the last page by a suitable amount.
                                  % This command does not take effect until the next page
                                  % so it should come on the page before the last. Make
                                  % sure that you do not shorten the textheight too much.

%%%%%%%%%%%%%%%%%%%%%%%%%%%%%%%%%%%%%%%%%%%%%%%%%%%%%%%%%%%%%%%%%%%%%%%%%%%%%%%%

%%%%%%%%%%%%%%%%%%%%%%%%%%%%%%%%%%%%%%%%%%%%%%%%%%%%%%%%%%%%%%%%%%%%%%%%%%%%%%%%

%%%%%%%%%%%%%%%%%%%%%%%%%%%%%%%%%%%%%%%%%%%%%%%%%%%%%%%%%%%%%%%%%%%%%%%%%%%%%%%%

\printbibliography

\end{document}